%% file: algo.tex
\documentclass[a4paper,aps,twocolumn,superscriptaddress,11pt,accepted=2021-09-27]{quantumarticle}
\pdfoutput=1

\input{algo_preamble}

\begin{document}

\title{Faster Born probability estimation via gate merging and frame optimisation}
\author{Nikolaos Koukoulekidis}
	\email{nk2314@imperial.ac.uk}
	\affiliation{Department of Physics, Imperial College London, London SW7 2AZ, UK}
\author{Hyukjoon Kwon}
	\affiliation{Department of Physics, Imperial College London, London SW7 2AZ, UK}
	\affiliation{Korea Institute for Advanced Study, Seoul, 02455, Korea}
\author{Hyejung H. Jee}
	\affiliation{Department of Computing, Imperial College London, London SW7 2AZ, UK}
\author{David Jennings}
	\affiliation{School of Physics and Astronomy, University of Leeds, Leeds, LS2 9JT, UK}
	\affiliation{Department of Physics, Imperial College London, London SW7 2AZ, UK}
\author{M. S. Kim}
	\affiliation{Department of Physics, Imperial College London, London SW7 2AZ, UK}
\date{\today}

\begin{abstract}
Outcome probability estimation via classical methods is an important task for validating quantum computing devices.
Outcome probabilities of any quantum circuit can be estimated using Monte Carlo sampling, where the amount of negativity present in the circuit frame representation quantifies the overhead on the number of samples required to achieve a certain precision. In this paper, we propose two classical sub-routines: circuit gate merging and frame optimisation, which optimise the circuit representation to reduce the sampling overhead.
We show that the runtimes of both sub-routines scale polynomially in circuit size and gate depth.
Our methods are applicable to general circuits, regardless of generating gate sets, qudit dimensions and the chosen frame representations for the circuit components. 
We numerically demonstrate that our methods provide improved scaling in the negativity overhead for all tested cases of random circuits with Clifford+$T$ and Haar-random gates, and that the performance of our methods compares favourably with prior quasi-probability simulators as the number of non-Clifford gates increases.
\end{abstract}
\maketitle

\section{Introduction}

Quantum computers promise to outperform their classical counterparts~\cite{feynman_simulating_1982, preskill2021quantum}.
However, the exact boundary between quantum and classical computational power is far from being fully characterised yet~\cite{Jozsa2014Classical, koh2017, aaronson2017, Hebenstreit2020, Yoganathan2019}.
Several works have demonstrated the difficulty in simulating certain quantum processes classically~\cite{aaronson2010computational, Bremner2011, Morimae2014, Bremner2016, Gao2017, Bermejo2018, fujii2018, bouland2018, boixo_characterizing_2018, Pashayan2020estimation, Yoganathan2019}.
Such results hint towards the ingredients that may be sufficient to achieve quantum advantage.
It is also possible to approach the boundary from the other side, namely by finding efficient methods to classically simulate families of quantum circuits~\cite{Valiant2002,Terhal2002,Bartlett2002,Aaronson2004,Gross2009,Brod2016,haug2022}, thereby providing insights on what ingredients are necessary for quantum advantage.

The question of efficient probability estimation has recently received vast attention due to the ongoing rapid development of quantum devices aiming to supersede classical capabilities (e.g.~\cite{Zhong2020,arute_quantum_2019,Bernien2017,Neill2018}).
Aided by powerful error mitigation techniques~\cite{bonilla_ataides_xzzx_2021,Bravyi_2021,Kandala_2019,Endo_2018,Temme_2017}, noisy intermediate-scale quantum (NISQ)~\cite{preskill_nisq} devices aim to deliver computational advantages, therefore fast and accurate outcome probability estimation is a necessity for quantitative benchmarking of the devices~\cite{harrow_quantum_2017, Temme_2017, Bennick_2017}.
For example, Google's recent experimental realisation of a quantum speed-up~\cite{arute_quantum_2019} relies on classical estimation methods to predict statistical features of the outcome probabilities.

It is expected that exact classical simulation of arbitrary quantum systems is inefficient, as the resource overhead exponentially grows with the size of the system.
Nevertheless, there are restricted classes of quantum circuits for which exact classical simulation is possible~\cite{cit:veitch}. 
The most notable example is given by circuits composed only with stabilizer states and gates in the Clifford group, which can be efficiently simulated classically via the Gottesman-Knill theorem~\cite{Gottesman_thesis}.

Probability estimation methods are varied and aim to explore the efficiency of circuit sampling or simulation beyond the regime of quantum circuits that admit tractable classical representations.
The estimation methods we mention in our work can be classified under one of two leading approaches~\cite{pashayan2021fast, Seddon2021}.
The first involves stabilizer rank-based simulators~\cite{Bravyi_2016,Bravyi2019,Bravyi2016,Qassim2019,huang2019approximate,Kocia2021,kissinger2022classical}, which rely on approximating the circuit components by stabilizer operators.
Every state or operation is assigned an exact or approximate stabilizer rank~\cite{Bravyi2016} indicating the number of stabilizer operators required to perform an exact or approximate decomposition of that component.
If a circuit component is non-classical, its stabilizer rank grows large thus inducing an exponential runtime cost for the estimation of the outcome probability.
Algorithms based on stabilizer decompositions have been very successful in estimating outcome probabilities of circuits dominated by Clifford gates and supplemented by a few types of magic states~\cite{Bravyi_2016,Bravyi2019,pashayan2021fast}.
These Clifford simulators are generalised from pure to noisy settings by the recently proposed density-operator stabilizer-rank simulator~\cite{Seddon2021}.
Furthermore, computing the stabilizer rank of arbitrary gates appears to be an intractable problem in the general case, so recent improvements on computing stabilizer rank bounds for specific non-Clifford states enhance runtimes significantly~\cite{Qassim2021}.

The other family of estimation methods relies on quasi-probabilistic representations of circuit components~\cite{Pashayan2015, Seddon2021, Bennick_2017, Stahlke_2014, Rall_2019, Seddon_2019} and such methods are in principle directly applicable to any quantum circuit without the need for state decompositions, in particular circuits with induced noise.
They are based on the notion of a \emph{frame representation} for the components of the circuit~\cite{ferrie2008frame,Ferrie_2009}.
Specifically, all components are represented by quasi-probability distributions in a certain frame and sampling on these distributions can be performed.
Since any state or gate admits such a representation, quasi-probability simulators naturally apply to arbitrary circuits with noise.
Many such frames have been studied~\cite{ferrie2008frame, Ferrie_2009, cit:gross, Ruzzi_2005, Marchiolli_2005, Rall_2019, Franca_2021}, and the runtime depends on the total negativity that is present in the circuit representation~\cite{Pashayan2015}.
A notable frame simulator is the dyadic frame simulator~\cite{Seddon2021} which relies on operator decompositions into stabilizer dyads $\ket{L}\bra{R}$, where $\ket{L}$ and $\ket{R}$ are pure stabilizer states. 
This method assigns dyadic negativity to non-classical elements, which quantifies the extent to which the operator's optimal linear decomposition into stabilizer dyads departs from a convex combination.
The dyadic simulator is a state-of-the-art quasi-probability frame simulator for qubits, as demonstrated by its low runtime scaling $O(4^{0.228t})$ with $t$ non-Clifford gates~\cite{Seddon2021}.
However, optimising the decomposition of an operator in dyads is computationally challenging.

Stabilizer rank simulators generally offer two advantages over estimation methods based on frame representations.
Firstly, they can be used for sampling the circuit output probabilities, which can be viewed as a stronger notion of simulation than probability estimation. 
Frame representation methods produce probability estimates with additive precision, which does not suffice for sampling~\cite{Bravyi2019,Pashayan2020estimation}.
Secondly, the stabilizer-rank algorithms developed in~\cite{Bravyi2019,Seddon2021,Bravyi_2016} are quadratically faster as they achieve a scaling of $O(2^{0.228t})$ in the asymptotic limit.
However, specialised simulators (e.g.~\cite{Bravyi2019,Bravyi_2016}) suffer from additional polynomial runtime factors, which tend to be more significant compared to the exponential runtime for the experimentally relevant case of big circuits with a low number of non-Clifford elements. 
Recently, an algorithm of additive precision~\cite{pashayan2021fast} has also been shown to asymptotically outperform the methods of~\cite{Bravyi2019,Bravyi_2016}, at least in certain parameter regimes.

In this paper, we focus on quasi-probability estimation methods based on frame representations and look for a way to improve the performance of outcome probability estimation. 
Recently, there has been a proposal of a Monte Carlo sampling algorithm which allows for quasi-probability estimation of circuits that contain a bounded amount of negativity in their representation~\cite{Pashayan2015}.
For classes of circuits in which negativity grows only polynomially in the number of input states, this estimation algorithm is efficient.
The negativity of the circuit therefore indicates the hardness of the probability sampling problem.
Although the negativity scales exponentially with the number of non-Clifford gates, the scaling factors hugely depend on the frame choice. 
Until now, however, the same fixed representation has been applied on every circuit component and the flexibility on reducing negativity has been limited.

Our aim is to explore the extent to which varying the frame representations of the components in a given circuit can lead to a reduction in the total circuit negativity.
To this end, we propose a pre-processing routine for any general quantum circuit, which aims at reducing the negativity overhead required for probability estimation.
Our proposed routine consists of two distinct sub-routines:
\begin{enumerate}
    \item \textbf{Circuit gate merging:} We introduce the idea of merging gates together into new $n$-qudit gates for fixed $n$ in the context of reducing sampling overhead.
    This sub-routine reduces the negativity of the entire circuit and is independent of the estimation method used.
    
    We demonstrate numerically that the average negativity reduction over a random ensemble of circuits is greater as the number of non-Clifford elements, e.g., $T$ gates, increases and is comparable to recent asymptotic negativity bounds~\cite{Seddon2021, howard_2017, Heinrich2019}. 
    Our routine does not depend on the specifics of the circuit gate set and can therefore be used in cases of gates which are hard to decompose, e.g., Haar-random gates.
    \item \textbf{Frame optimisation:} We introduce the idea of using different frames to represent the input and output phase spaces of the gates in the circuit.
    This is inspired by work in continuous variables~\cite{Keshari2016}, but our approach is novel in the context of discrete quasi-probability sampling methods.
    
    We argue that this sub-routine compliments gate merging as an additional source of negativity reduction when merging is no longer efficient. 
    We then demonstrate numerically that instances of Clifford+$T$ circuits and circuits with Haar-random gates admit significant negativity reductions by introducing additional frames in the circuit representation.
\end{enumerate}
We note that a polynomial runtime for these classical sub-routines with respect to the circuit size should be guaranteed to effectively reduce the overall runtime of the sampling method. 
As proof of principle, we provide explicit algorithms in the main text that ensure this condition for each sub-routine.

This paper is organised as follows.
In Section~\ref{sec:outline}, we review the frame representation and the estimation algorithm using quasi-probability representations of a given quantum circuit. 
In Section~\ref{sec:overview}, we outline our results within the context of the current state of quasi-probability simulator research. 
In Section~\ref{sec:merge} and~\ref{sec:opt} we describe the two sub-routines in more detail, before providing a summary in Section~\ref{sec:summary}.

\section{Preliminaries}
\label{sec:outline}

\subsection{Frame representation of quantum circuits}
We first give a brief overview on classical circuit sampling based on the method of frame representation. 
Suppose that an $N$-qudit quantum circuit $C$ is composed of the initial state preparation $\rho$, sequential quantum gates $U_1, U_2, \dots ,U_L$ and the measurement effect $E$.
The outcome probability of the quantum circuit $p_C = {\rm Tr} [ U_L \dots U_2 U_1 \rho U_1^\dagger U_2^\dagger \dots U_L^\dagger E]$ can be estimated by describing the quantum state $\rho$ as quasi-probability distributions over phase space points $\blambda \in \mathbb{Z}^{2N}_d$ and the quantum operations $U_i$ as the transition matrices of the distributions. 
More specifically, a phase space can be constructed from a frame defined as a set of operators ${\cal F} \coloneqq \{ F(\blambda)\}$ and its dual ${ \cal G} \coloneqq \{ G(\blambda) \}$~\cite{ferrie2008frame,Ferrie_2009}, such that any operator $O$ is expressed as
\begin{equation}
    O = \sum_\blambda {\rm Tr}[F(\blambda) O ] G(\blambda).
\end{equation}
For a given frame, the outcome probability can be expressed in terms of the representation as
\begin{equation}
    p_C = \sum_{\blambda_0, \dots, \blambda_L} W_E(\blambda_L) \left[\prod_{l=1}^L  W_{U_l}(\blambda_l|\blambda_{l-1})\right]  W_\rho(\blambda_0),
\end{equation}
where we define
\begin{align}
    W_\rho(\blambda) &= {\rm Tr}[F(\blambda) \rho],\\
    W_U(\blambda'|\blambda) &={\rm Tr}[ F(\blambda')U G(\blambda)U^\dagger], \text{ and }\\
    W_E(\blambda) &={\rm Tr}[E G(\blambda)].
\end{align}
In the case where 1) $\rho$ and $E$ are products of local initial states and measurement effects, 2) $W_\rho(\blambda_0)$ and $W_E(\blambda_L)$ are classical probability distributions, and 3) $W_{U_l}(\blambda_l|\blambda_{l-1}),$ for $\ell=1,\dots,L$, are classical conditional probability distributions for all $l$, efficient classical simulation is possible, where the sampling runtime scales polynomially with $N$ and $L$~\cite{Mari2012}.
The simulation is performed by sampling the trajectories of $(\blambda_0, \dots, \blambda_L)$ from the initial distribution $P(\blambda_0) = W_\rho(\blambda_0)$ and the transition matrix at each step $P_l(\blambda_l|\blambda_{l-1}) = W_{U_l}(\blambda_l|\blambda_{l-1})$, which leads to the probability estimate, $\hat{p}_C = W_E\left(\blambda_L\right)$.
Taking an average over $M$ probability estimates converges to the Born probability as $M$ increases.

\subsection{Overhead of classical simulation}
Non-classicality in the quantum process is represented by negativities in $W_\rho(\blambda)$ or $W_U(\blambda'|\blambda)$, which gives rise to quasi-probabilities. 
In general, $W_\rho(\blambda)$ and $W_U(\blambda'|\blambda)$ consist of real components that can attain negative values, while satisfying the normalisation conditions,
\begin{align}
    \sum_{\blambda \in \mathbb{Z}^{2N}_d} W_\rho(\blambda) &= 1 \text{ and} \\
    \sum_{\blambda' \in \mathbb{Z}^{2N}_d} W_U(\blambda'|\blambda) &= 1 \text{ for all } \blambda \in \mathbb{Z}^{2N}_d.
\end{align}
Despite the presence of negativities in the distributions and update matrices, Monte Carlo methods can still be used with adjustments as introduced by Pashayan~\emph{et al.}~\cite{Pashayan2015} in order to perform probability sampling.
This can be done by sampling over $P(\blambda_0) = |W_\rho(\blambda_0)|/\sum_{\blambda_0} |W_\rho(\blambda_0)|$ for the initial state preparation and taking the transition matrix of $P_l(\blambda_l|\blambda_{l-1}) = |W_{U_l}(\blambda_l|\blambda_{l-1})|/\sum_{\blambda_l} |W_{U_l}(\blambda_l|\blambda_{l-1})|$ for the quantum gate, while keep track of the signs.
In this case, the probability estimate is modified to
\begin{equation}
\begin{split}
\hat{p}_C 
=\ &\sign \left( W_\rho(\blambda_0) \prod_{l=1}^L W_{U_l}(\blambda_l|\blambda_{l-1})\right) \times \\
&N_\rho \left(\prod_{l=1}^L N_{U_l}(\blambda_{l-1}) \right) W_E(\blambda_L),    
\end{split}
\end{equation}
where we have defined
\begin{align}
    N_\rho &\coloneqq \sum_{\blambda_0} \left| W_\rho(\blambda_0) \right |\\
    N_{U_l}(\blambda_{l-1}) &\coloneqq \sum_{\blambda_l} \left| W_{U_l}(\blambda_l|\blambda_{l-1}) \right|.
\end{align}
In order to converge to the Born probability, one can similarly take the average of increasingly many probability estimates sampled over trajectories $(\blambda_0, \dots, \blambda_L)$ using distributions $P(\blambda_0)$ and $P_l(\blambda_l|\blambda_{l-1})$.

This directly relates the total amount of circuit negativity to the computational overhead: the larger the negativity in the circuit, the more samples required for an accurate estimation.

\begin{observation}[Pashayan \textit{et al.}~\cite{Pashayan2015}]
\label{obs:Pashayan}
The outcome probability $p_C$ of the quantum circuit $C$ can be estimated by $\hat{p}_C$ from the number of samples
\begin{align}\label{eq:n_samples_in_negativity}
    M \geq M(\epsilon, \delta) = \frac{2}{\epsilon^2} N_C^2 \ln (2/\delta),
\end{align}
with at least probability $1 -\delta $ of having error less than $\epsilon$. Here, 
\begin{equation}\label{eq:circuit_negativity}
    N_C = N_\rho \times \left[ \prod_{l=1}^L \max_{\blambda_0,\dots,\blambda_{L-1}} N_{U_l}(\blambda_{l-1}) \right] \times \max_{\blambda_L} \left | W_E(\blambda_L) \right|,
\end{equation}
is the (maximum) circuit negativity.
\end{observation}
As is clear by Eq.~(\ref{eq:n_samples_in_negativity}), the negativity of the circuit acts as an overhead for the convergence time of the sampling algorithm, therefore it is desirable to reduce it before executing the sampling by considering different frame choices.

\section{Main results}
\label{sec:overview}

In this work, we develop a pre-processing routine to reduce the negativity of the circuit, which in turn reduces the number of samples required to estimate the outcome probability of the circuit. Our routine is applicable to any general circuit consisting of a product input state and product measurement, but independently of the input state dimension, the gate set (e.g. Clifford unitaries or Haar-random gates) and adaptive operations based on intermediate measurement outcomes.

\begin{figure*}
    \includegraphics[width=1.\linewidth]{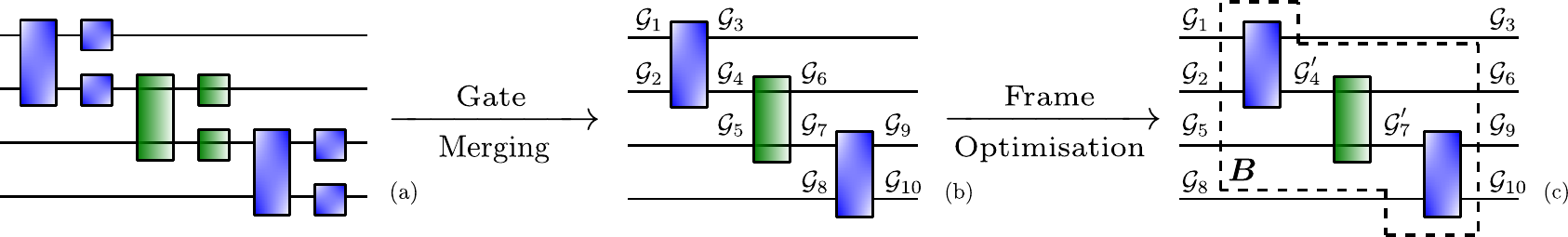}
    \centering
    \caption{Sketch of routine on a toy circuit. 
    The first step (a) $\rightarrow$ (b) is gate merging, here implemented with $n = 2$. 
    Gates that share input and output wires merge in the schematic way depicted in the figure.
    The second step (b) $\rightarrow$ (c) is frame optimisation, here implemented with $\ell = 2$ in a block $\mathbf{B}$ comprising the three merged gates.
    The optimisation results into updated frames ${\cal G}_4', {\cal G}_7'$, while the remaining frames that connect the block $\mathbf{B}$ to the rest of the circuit components are left unchanged at this optimisation cycle.
    }
    \label{fig:sketch_routine}
\end{figure*}

\subsection{Frame parametrisation}\label{subsec:frame_parametrisation}

The central focus of our work is to consider frame parametrisations that are allowed to vary across the circuit. It is clear from the definitions that the circuit negativity of a given circuit in Eq.~\eqref{eq:circuit_negativity} depends on the choice of the frame $\mathcal{F}$ and its dual $\mathcal{G}$.
Note that $\mathcal{F}$ and $\mathcal{G}$ are uniquely defined by each other for phase space dimension equal to $d^2$, where $d$ is the qudit dimension.
We therefore make the dependence clear by labelling the representation functions by  $\mathcal{G}$:
\begin{align}
    W^{\mathcal{G}}_\rho(\blambda) &= {\rm Tr}[F(\blambda)\rho], \label{eq:QPD_with_frame_para_state}\\
    W^{\mathcal{G}'|\mathcal{G}}_U(\blambda'|\blambda) &={\rm Tr}[ F'(\blambda') U G(\blambda)U^\dagger], \text{ and } \label{eq:QPD_with_frame_para_gate}\\
    W^{\mathcal{G}}_E(\blambda) &={\rm Tr}[E G(\blambda)], \label{eq:QPD_with_frame_para_meas}
\end{align}
where we used different frames, $\mathcal{G}$ and $\mathcal{G}'$, for the input and output wires respectively in the definition of $W_{U}^{\mathcal{G}'|\mathcal{G}}$.

In order to ensure that the number of frames does not grow exponentially with the number of qudits $N$, we restrict to \emph{product} frames that are constructed as tensor products of single qudit frames. 
This allows us to parametrise each single qudit phase space separately, rather than the entire $N$-qudit phase space.
Therefore, we reserve the label ${\cal G}$ for denoting single qudit frames and the boldface symbol ${\boldsymbol{\cal G}}$ for denoting a set of single qudit frames.
The negativity of each circuit component can now be expressed as
\begin{align}
    N^{{\boldsymbol{\cal G}}}_\rho &= \sum_\blambda \left| W^{{\boldsymbol{\cal G}}}_\rho(\blambda) \right |,\\
    N^{{\boldsymbol{\cal G}}'|{\boldsymbol{\cal G}}} &= \max_\blambda \left[ \sum_{\blambda'} \left| W^{{\boldsymbol{\cal G}}'|{\boldsymbol{\cal G}}}_U(\blambda'|\blambda) \right| \right], \text{ and } \\
    N^{{\boldsymbol{\cal G}}}_E &= \max_{\blambda} \left| W^{{\boldsymbol{\cal G}}}_E(\blambda) \right |,
\end{align}
where ${\boldsymbol{\cal G}}, {\boldsymbol{\cal G}}'$ contain elements from the complete set of frames required to represent the circuit.
In practice, each circuit component is parametrised only via the frames that correspond to its input and output wires.
For example, in Fig.\ref{fig:sketch_routine}(b), when parametrising the first gate in the sequence, we can simply consider ${\boldsymbol{\cal G}}$ as the set $\{ {\cal G}_1, {\cal G}_{2} \}$ and ${\boldsymbol{\cal G}}'$ as the set $\{ {\cal G}_3, {\cal G}_4 \}$.
If only a unique frame representation ${\cal G}$ is used for all circuit components, then ${\boldsymbol{\cal G}} = {\boldsymbol{\cal G}}' = \{{\cal G}\}$ in the expressions above and the label can be dropped, simplifying to the notation of the previous section.
The total negativity of the parametrised circuit can now be expressed as a function of the circuit frame set $\boldsymbol{ \cal G}$:
\begin{equation}\label{eq:circuit_param_negativity}
    N_C({\boldsymbol{\cal G}}) = N_\rho^{{\boldsymbol{\cal G}}} \times \left[ \prod_{l=1}^L N_{U_l}^{{\boldsymbol{\cal G}}'|{\boldsymbol{\cal G}}} \right] \times N_E^{{\boldsymbol{\cal G}}}.
\end{equation}
We note that Observation~\ref{obs:Pashayan} still holds by replacing $N_C$ with the more general form $N_C(\boldsymbol{\mathcal{G}})$. Our main objective is to study the reduction of this circuit negativity by tuning $\boldsymbol{ \cal G}$.

\subsection{Examples of frame parametrisations}
\label{subsec:examples_frame_para}

While our results are general and applicable to any family of parametrised frames, in this work we provide two examples of explicit, product frame parametrisations:
(i) parametrised Wigner frames and (ii) rotated Pauli frames. 

\emph{Parametrised Wigner frames} employ the conventional phase space of the discrete Wigner function~\cite{Ferrie_2009, cit:gross}. Let us define the discrete displacement operator for a $d$-dimensional system as
\begin{equation}\label{eq:displacement_operators}
    D(p,q)= \chi(-2^{-1} pq) Z^p X^q,    
\end{equation}
where $\chi(q) = e^{i(2\pi/d)q}$.
For a qubit system ($d=2$), this takes the form $D(p,q)=i^{pq}Z^p X^q$. 
It can be generalised to an $N$-qudit system as
\begin{equation}
    D(\blambda) = \bigotimes_{i=1}^N D(p_i, q_i)\,,
\end{equation}
where $\blambda \coloneqq (p_1, q_1, p_2, q_2, \dots, p_N, q_N)^T \in \mathbb{Z}^{2N}_{d}$ denotes a phase space point of the whole system. We then define the frame ${\cal F} = \{ F(\blambda) \} \coloneqq \{ D(\blambda) F_0 D^\dagger (\blambda) \}$ and its dual frame ${\cal G} = \{ G(\blambda) \} \coloneqq \{ D(\blambda) G_0 D^\dagger (\blambda)\}$ using the following reference operators:
\begin{align}\label{eq:para_Wigner_frame_referece_operators}
    F_0 &= \frac{1}{d} \sum_\blambda \left[ \frac{1}{g(\blambda)} \right] D(\blambda)\\
    G_0 &= \frac{1}{d} \sum_\blambda  g(\blambda) D(-\blambda)\,,
\end{align}
where we introduced the parametrisation function $g(\blambda)$. Note that the following relation holds:
\begin{align}
    g(\blambda) = \tr\left[G_0 D(\blambda)\right]\,,
\end{align}
so the parametrisation function $g(\blambda): \mathbb{Z}^{2N}_d \mapsto \mathbb{C}\setminus\{0\}$ can be fully characterised by the reference operator $G_0$.
In order to impose that $W_\rho^{\cal G}(\blambda)$ is real-valued and that $\sum_\blambda W_\rho^{\cal G}(\blambda) = 1$, we need the additional conditions $g^*(\bomega) = g(-\bomega)$ and $g(\boldsymbol{0}) = 1$, which are equivalent to $G_0^\dagger = G_0$ and ${\rm Tr} [G_0] = 1$ respectively.
By taking $g(\blambda)=1$ for all $\blambda$, the conventional discrete Wigner function~\cite{cit:gross} is recovered. One can calculate the quasi-probability distributions of circuit elements via Eq.~\eqref{eq:QPD_with_frame_para_state}-\eqref{eq:QPD_with_frame_para_meas} using the defined frame and dual frame. In odd dimensions, the parametrised Wigner frame is a good choice for Clifford dominated circuits as Clifford gates do not possess any negativity in the conventional Wigner distribution. Therefore, $g(\blambda)=1$ is already optimal for most circuit elements when considered in isolation and constitutes an obvious starting point for frame optimisation.

In the qubit case, it is known that the Hadamard and CNOT gates have non-zero negativity even in the conventional Wigner distribution \cite{raussendorf2017contextuality}, which motivates us to introduce the next frame parametrisation, valid only for qubits: the rotated Pauli frames.

\emph{Rotated Pauli frames} are based on the Bloch decomposition of a quantum operator. Consider the set of displacement operators for a single qubit $\{D(\blambda)\}$ as defined in Eq.~\eqref{eq:displacement_operators} for $\blambda \in \mathbb{Z}_2^2=\{(0,0),(0,1),(1,0),(1,1)\}$. The usual Bloch vector for a single-qubit state $\rho$ can be written as
\begin{equation}\label{eq:bloch_vector}
    W_{\rho}(\blambda) = \frac{1}{2}\tr\left[ \rho D(\blambda) \right],
\end{equation}
and this defines a valid quasi-probability distribution with frame $\{\frac{1}{2}D(\blambda)\}$. We can define a new frame by applying a rotation to the space of the Bloch vector. Let us consider a rotational angle vector $\bmth := (\theta_X, \theta_Y, \theta_Z)$ and a corresponding rotation operator $R(\bmth) := R(\theta_Z)R(\theta_Y)R(\theta_X)$, where $R(\theta_X) \coloneqq e^{-i\theta X / 2}$ and similarly for $Y,Z$. Applying this to the Bloch vector in Eq.~\eqref{eq:bloch_vector} results in a set of rotated displacement operators, parametrised by $\bmth$:
\begin{widetext}
\begin{align}
    D^{\bmth}(0,0) &\coloneqq \frac{1}{2} \id \\
    D^{\bmth}(0,1) &\coloneqq \begin{pmatrix} -\sin{\theta_Y} & e^{-i \theta_X} \cos{\theta_Y} \\ e^{+i \theta_X} \cos{\theta_Y} & \sin{\theta_Y} \end{pmatrix} \\
    D^{\bmth}(1,0) &\coloneqq \begin{pmatrix} \cos{\theta_Y}\cos{\theta_Z} & e^{-i \theta_X} (\sin{\theta_Y}\cos{\theta_Z} + i\sin{\theta_Z}) \\ e^{+i \theta_X} (\sin{\theta_Y}\cos{\theta_Z} - i\sin{\theta_Z}) & -\cos{\theta_Y}\cos{\theta_Z} \end{pmatrix} \\
    D^{\bmth}(1,1) &\coloneqq \begin{pmatrix} \cos{\theta_Y}\sin{\theta_Z} & e^{-i \theta_X} (\sin{\theta_Y}\sin{\theta_Z} + i\cos{\theta_Z}) \\ e^{+i \theta_X} (\sin{\theta_Y}\sin{\theta_Z} - i\cos{\theta_Z}) & -\cos{\theta_Y}\sin{\theta_Z} \end{pmatrix}\,.
\end{align}
\end{widetext}
Then, we define the frame ${\cal F} = \{ F(\blambda) \}$ and its dual frame ${\cal G} = \{ G(\blambda) \}$ as
\begin{align}
    F(\blambda) &\coloneqq \frac{1}{2^N} D^{\bmth}(\blambda), \\
    G(\blambda) &\coloneqq D^{\bmth}(\blambda)\,,
\end{align}
which provide a parametrised frame representation for a qubit.
This can be generalised to an $N$-qubit system via
\begin{equation}
    D^{\bmth}(\blambda) = \bigotimes_{i=1}^N D^{\bmth_i}(\blambda_i)\,
\end{equation}
with $\bmth = (\bmth_1, \bmth_2, \dots, \bmth_N$), where $\bmth_i$ is the rotational angle vector for the $i$-th qubit. 
The rotated Pauli frames possess the desired property that all stabilizer states and Clifford gates have zero negativity in the conventional Bloch frame representation with $\bmth=(0,0,0)$. Thus, when a given qubit circuit is dominated by Clifford gates, it can be advantageous to employ the rotated Pauli frame.

\subsection{Pre-processing routine for negativity reduction}

The central idea of our pre-processing routine for negativity reduction can now be expressed by the following lower bounds on gate negativity.

\begin{theorem}\label{thm:MainStatement}
For two consecutive gates $U$ and $V$, 
the following bounds on negativity hold:
\begin{equation}\label{eq:MainStatement}
    N^{\boldsymbol{\cal G}|\boldsymbol{\cal G}}_V N^{\boldsymbol{\cal G}|\boldsymbol{\cal G}}_U \;\geq\; \min_{\boldsymbol{\cal G}'}\; N^{\boldsymbol{\cal G}|\boldsymbol{\cal G}'}_V N^{\boldsymbol{\cal G}'|\boldsymbol{\cal G}}_U \;\geq\; N^{\boldsymbol{\cal G}|\boldsymbol{\cal G}}_{VU},
\end{equation}
where $\boldsymbol{\cal G}$ and $\boldsymbol{\cal G}'$ are frame sets that represent gates $U,V$ and $UV$.
\end{theorem}
\begin{proof}
    The first inequality holds since $\boldsymbol{\cal G}$ is one specific choice of the optimisation variable set $\boldsymbol{\cal G}'$. The second inequality is due to Observation~\ref{obs:merge_gates} in the next section.
\end{proof}
\noindent Theorem~\ref{thm:MainStatement} motivates us to introduce two sub-routines applicable to any quasi-probability estimation algorithm with runtime cost determined by the circuit negativity.

The second inequality in Theorem~\ref{thm:MainStatement} suggests that merging two gates into one is generally advantageous in minimising the total negativity. This leads to the first pre-processing sub-routine, \textit{gate merging}. The inequality is independent of the specific frame parametrisation and can be directly extended to an arbitrary number of gates. The trade-off is that the merged gate may be of a larger size. For example, if $U$ and $V$ are 2-qudit gates sharing one wire between them, gate $VU$ will be a 3-qudit gate. The dimension of the merged gate increases exponentially as the number of qudits involved becomes larger, hence one should truncate the maximum number of qudits acted on by the merged gates, which we define as the spatial parameter $n$.

The first inequality in Theorem~\ref{thm:MainStatement} states that, unless the frames between two gates in sequence are already optimal, we can always reduce the total negativity of the two gates by optimising the frames they share.
This leads to the second sub-routine, \textit{frame optimisation}.
The optimisation can be directly generalised to a circuit block $\mathbf{B}$ containing a sequence of $\ell$ frames $\boldsymbol{\calG}$ by simultaneously optimising all the frames in the block, $\min_{\boldsymbol{\calG}} N_{\mathbf{B}}(\boldsymbol{\calG})$.
The temporal parameter $\ell$ is the number of frames to be optimised in one optimisation cycle.
The optimisation takes place iteratively in the sense that every optimisation cycle optimises the frames within a block, taking as an initial state the optimised frames obtained from the previous cycle.
This ensures that negativity cannot increase above its initial value, no matter how many optimisation cycles occur.

Given fixed values for the truncation parameters $n,\ell$, we show in the following two sections that the total runtime $\tau$ of our routine is polynomial in the number of circuit components,
\begin{equation}
    \tau = O(N, L^2).
\end{equation}
In general, larger $n$ or $\ell$ give larger negativity reduction at the cost of additional classical computation.

We note that gate merging yields lower negativity than any frame optimisation between the gates. 
However, fixing $n < N$ prevents us from merging gates indefinitely, so frame optimisation can then be used for further negativity reduction.

\begin{algorithm}[H]
\caption{Outcome Probability Estimation with Merging and Optimisation}
\label{alg:complete_algorithm}
\begin{algorithmic}[1]
\Statex \textbf{Input:} An $N$-qudit quantum circuit $C$ with a product input state $\rho=\rho_1\otimes\dots\otimes\rho_N$, the list of gates ${\cal U} = \{U_1,...,U_L\}$, and the product measurement operator $E=E_1\otimes\dots\otimes E_N$; the spatial parameter $n$; the temporal parameter $\ell$; the desired accuracy $\epsilon$.
\State Run gate merging (Sub-routine~\ref{alg:merge_algorithm}) with the input gate sequence and $n$ and return the merged gate sequence $\{V_1,...,V_{L'}\}$ with $L'\leq L$ consisting of gates acting on at most $n$ qudits.
\State Run frame optimisation (Sub-routine~\ref{alg:frame_opt}) with the merged circuit and $\ell$ and return the optimised frame sequence $\boldsymbol{\mathcal{G}}_{\rm opt}$.
\State Run a sampling algorithm to achieve the input accuracy $\epsilon$ according to Eq.~\eqref{eq:n_samples_in_negativity} using the quasi-probability representations of the merged circuit obtained with the optimised frame sequence $\boldsymbol{\mathcal{G}}_{\rm opt}$.
\Statex \textbf{Output:} $p_{\rm est}$, the estimated outcome probability.
\end{algorithmic}
\end{algorithm}

We present an algorithm for Born probability estimation, including our complete pre-processing routine and sampling, in Algorithm~\ref{alg:complete_algorithm} and illustrate its implementation on a toy circuit in Fig.~\ref{fig:sketch_routine}. 
In the following two sections, we discuss in more detail how the two sub-routines, gate merging and frame optimisation, can be implemented.
For clarity, we focus on qubit circuits and on the frame parametrisations introduced in the previous section, although our methods are general.

\section{Gate merging}
\label{sec:merge}

The central idea of our first sub-routine, gate merging, is that the sampling cost of a merged circuit block consisting of multiple quantum gates is in general lower than sequential sampling of each gate. More precisely, this can be summarised as the following observation:

\begin{observation}\label{obs:merge_gates}
    Let $\{U_1, U_2, \dots, U_k\}$ be a sequence of quantum gates. The negativity of the merged gate $U = U_k \dots U_2 U_1$ is always less or equal to the product of the individual negativities, i.e.,
    \begin{equation}
        N_{U}^{\boldsymbol{\cal G}} \leq \prod_{i=1}^k N_{U_i}^{\boldsymbol{\cal G}},
        \label{eq:merged_gates}
    \end{equation}
    for any frame set $\boldsymbol{\cal G}$ assigned to the gate sequence.
\end{observation}
\begin{proof}
    It is sufficient to prove the statement for two gates $U$ and $V$.
    By noting that the quasi-probability of the merged gate is expressed as
    \begin{equation}
        W_{VU}^{\boldsymbol{\cal G}}(\blambda_3 | \blambda_1) = \sum_{\blambda_2} W_{V}^{\boldsymbol{\cal G}}(\blambda_3 | \blambda_2) W_{U}^{\boldsymbol{\cal G}}(\blambda_2 | \blambda_1),
    \end{equation}

    the negativity of the gate can be bounded as
    \begin{equation}
    \begin{aligned}
        N_{VU}^{\boldsymbol{\cal G}} &= \max_{\blambda_1}{\sum_{\blambda_3} \abs{W_{VU}^{\boldsymbol{\cal G}}(\blambda_3 | \blambda_1)}} \\
        &= \max_{\blambda_1}{\sum_{\blambda_3} \abs{\sum_{\blambda_2} W_{V}^{\boldsymbol{\cal G}}(\blambda_3 | \blambda_2) W_{U}^{\boldsymbol{\cal G}}(\blambda_2 | \blambda_1)}} \\
        &\leq \max_{\blambda_1}{\sum_{\blambda_2} \abs{W_{U}^{\boldsymbol{\cal G}}(\blambda_2 | \blambda_1)} \sum_{\blambda_3} \abs{W_{V}^{\boldsymbol{\cal G}}(\blambda_3 | \blambda_2)}} \\
        &\leq N_{U}^{\boldsymbol{\cal G}} \max_{\blambda_2}{\sum_{\blambda_3} \abs{W_{V}^{\boldsymbol{\cal G}}(\blambda_3 | \blambda_2)}} \\
        &= N_{V}^{\boldsymbol{\cal G}} N_{U}^{\boldsymbol{\cal G}}.
    \end{aligned}
    \end{equation}
    We then apply this argument iteratively to any sequence of quantum gates $\{ U_1, U_2, \dots, U_k\}$ to obtain Eq.~\eqref{eq:merged_gates}, which completes the proof.
\end{proof}

Such a negativity reduction can be exemplified by considering the Toffoli gate, which can be optimally decomposed into four $T$ gates~\cite{cody2013} along with Clifford gates and Pauli measurements.
We compare the negativity of the Toffoli gate itself and its decomposed gate sequence using the Pauli frame, where the negativity only comes from non-Clifford gates. One can readily observe that the Toffoli gate negativity $N_{\rm Toffoli}^{\rm Pauli} = 2$ is lower than the total negativity of the decomposed gate sequence $\left[ N_T^{\rm Pauli} \right]^4 = 4$.

The idea of reducing the negativity of quantum gates by merging (Eq.~\eqref{eq:merged_gates}) can be compared to the submultiplicativity of magic state negativity characterised by the robustness measure (${\cal R}$), which obeys ${\cal R}(\rho_1 \otimes \rho_2) \leq {\cal R}(\rho_1) {\cal R}(\rho_2)$  \cite{howard_2017}. In particular, the robustness of the $T$ state is equivalent to the negativity of the $T$ gate from the sampling cost viewpoint, as one $T$ gate can be ``gadgetised'' via Cliffords and a single $T$ state~\cite{Bravyi_2016}. 
In Ref.~\cite{howard_2017}, the asymptotic negativity per single $T$ gate is $\lim_{t \rightarrow \infty}\left[ {\cal R}\left(\ket{T}^{\otimes t}\right) \right]^{1/t} \approx 2^{0.272}$ which provides a lower bound on their sampling runtime $\Omega (4^{0.272t})$.

In order to compare this with the gate merging method, we consider an $n$-qubit block consisting of Clifford+$T$ gates (see Fig.~\ref{fig:merge_fig}(f) for an example with $n=5$). This can be compared to considering $n$ $T$ states in the robustness measure, having the same number of qubits (i.e., the size of Hilbert space) in the block to evaluate the negativity.
\begin{figure}[t]
    \includegraphics[width=1
    \linewidth]{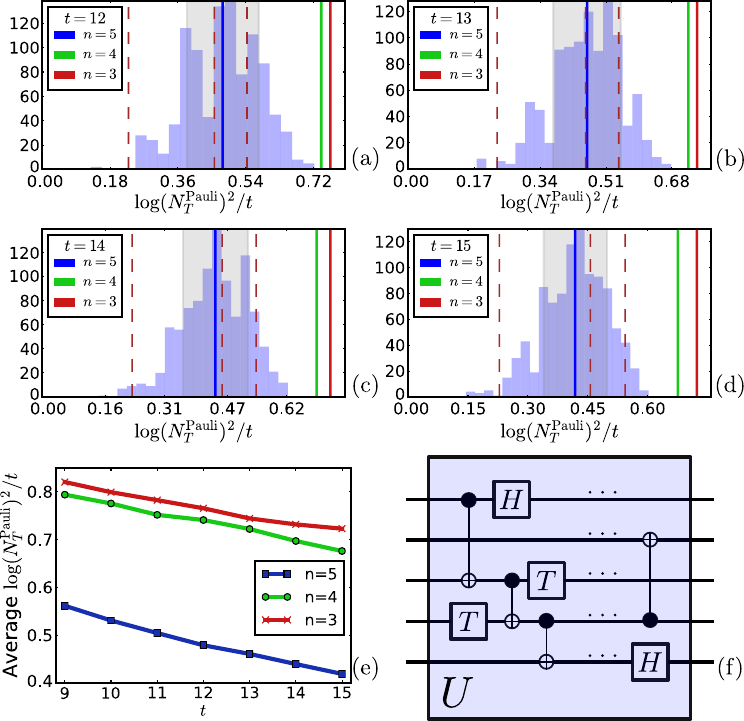}
    \caption{Histograms of $1000$ random Clifford$+T$ circuits with $N=5$ consisting of 100 1-qubit and 2-qubit Clifford gates, supplemented by $t$ $T$ gates and merged using spatial parameter $n=5$.
    The leftmost (blue) solid line with the gray region depict the average and standard deviation of each histogram.
    The brown and green solid lines (from right to left) represent the higher averages of the corresponding histograms for $n=3$ and $4$ respectively.
    Vertical dashed lines provide some state-of-the art scalings, more specifically from left to right: $O(2^{0.228t})$ of the Bravyi-Gosset algorithm from~\cite{Bravyi_2016} based on the stabilizer rank, $O(4^{0.228t})$ of the dyadic frame simulator from~\cite{Seddon2021} and the lower bound $\Omega (4^{0.272t})$ based on the robustness of magic from~\cite{howard_2017}.
    As $t$ increases, we observe a higher frequency of circuits with log negativity squared per $T$ gate lower than the robustness lower bound: (a) $71\%$, (b) $81\%$, (c) $89\%$, (d) $95\%$.\\
    (e) Histogram average for $n=3,4,5$ against $t$.\\
    (f) Example $5$-qubit merged gate $U$ made up from Clifford gate ($CNOT$s and $H$) and $T$ gates.
    }
    \label{fig:merge_fig}
\end{figure}
Figs.~\ref{fig:merge_fig}(a-d) show the distribution of the negativity of $1000$ random $n$-qubit blocks consisting of $100$ Clifford gates and $t$ $T$ gates.
We observe that the negativity per $T$ gate after merging the gate sequences in a random $n$-qubit block can be occasionally lower than the robustness measure of $n$ $T$ states~\cite{howard_2017}.
We also note that the negativity reduction works efficiently when the number of $T$ gate in the block, $t$, increases. For example, when $n=5$ and $t=15$, $95\%$ of the randomly chosen merged blocks yield negativity per $T$ state lower than the robustness measure. 
We also plot in Fig.~\ref{fig:merge_fig}(e) the average negativity per $T$ gate versus $t$, demonstrating that it is decreasing, which implies that our appoach can prove efficient when the structure of the gate block considered becomes more complicated.

The main advantage of our approach is that it is not limited to a particular type of gate set, e.g. Clifford+$T$ circuits, but can be directly applied to any types of quantum gates. The aforementioned approaches using stabilizer rank, robustness and generalised robustness rely on the gadgetisation of a non-Clifford gate using magic states. 
Therefore, evaluating the classical overhead should be preceded by finding an optimal Clifford gadget with minimum resource of magic. 
On the other hand, gate merging does not have such a limitation, so it can be useful when the efficient decomposition of a quantum circuit into non-stabilizer states and Clifford gates is non-trivial. 
We also highlight that merging gates reduces the negativity independently of the choice of frames.

We now describe the gate merging method for a generic $N$-qubit quantum circuit with $L$ gates. This can be done by grouping the quantum circuit into $n$-qubit blocks (see Fig.~\ref{fig:sketch_routine}(a)$\rightarrow$(b)), then Observation~\ref{obs:merge_gates} guarantees that the negativity of each block is reduced after merging the gate sequences in it. 
There are various ways of grouping the circuit into $n$-qubit blocks, but we introduce the iterative Sub-routine~\ref{alg:merge_algorithm} for concreteness.
The broad idea of the sub-routine is to iteratively connect any yet unmerged (disjoint) gates. 
All gates remain in the set ${\cal U}_{\rm disj}$ until they either finally act on $n$ qubits or cannot connect to other gates anymore, when they are move to the output set ${\cal U}_{\rm merged}$.

\begin{subroutine}[H]
\caption{Gate merging}
\label{alg:merge_algorithm}
\begin{algorithmic}[1]
\Statex \textbf{Input:} List of gates ${\cal U} = \{U_1,...,U_{L}\}$ in
\Statex \hspace{4pt} qudit quantum circuit $C$ and spatial
\Statex \hspace{4pt} parameter $n$.
\State Define list of merged gates ${\cal U}_{\rm merged}\gets \{\}$, and 
\Statex \hspace{4pt} list of disjoint gates ${\cal U}_{\rm disj} \gets \{\}$
\For{$U_i \in {\cal U}$} 
\State Set target gate $U_{\rm target} \gets U_i$
    \For{$V \in {\cal U}_{\rm disj}$}
       \If{$U_{\rm target}$ shares a wire with $V$}
        \Statex \hspace{45pt} Remove $V$ from ${\cal U}_{\rm disj}$.
            \If{${\rm rank}(U_{\rm target} V) > d^n$}
            \Statex \hspace{60pt} Add $V$ to ${\cal U}_{\rm merged}$.
            \ElsIf{{${\rm rank}(U_{\rm target} V) \leq d^n$}}
            \Statex \hspace{60pt} $U_{\rm target} \gets U_{\rm target} V$.
            \EndIf
        \EndIf
    \EndFor
    \State Add $U_{\rm target}$ to ${\cal U}_{\rm disj}$.
\EndFor
\For{
$U_i \in {\cal U}_{\rm merged}$}
    \State Set target gate $U_{\rm target} \gets U_i$
    \For{
$V \in {\cal U}_{\rm disj}$}
        \If{
${\rm rank}(U_{\rm target} V) \leq d^n$}
        \Statex \hspace{60pt} 
$U_{\rm target} \gets U_{\rm target} V$.
        \EndIf
    \EndFor
    \State Add $U_{\rm target}$ to ${\cal U}_{\rm disj}$.
\EndFor
\State Append ${\cal U}_{\rm disj}$ to ${\cal U}_{\rm merged}$.
\Statex \textbf{Output:} $\mathcal{U}_{\rm merged}$.
\end{algorithmic}
\end{subroutine}

At every step, a target gate $U_{\rm target}$, the algorithm searches through the disjoint gates to find the next one that is connected to $U_{\rm target}$. We therefore require to search less than $L$ gates for every $U_{\rm target}$, while the cost of merging two gates (i.e., multiplying) is $O(2^{2n})$, a constant as we fix $n < N$. So the full gate merging sub-routine scales as $O(2^{2n}L^2)$. The computational cost to compute the transition matrix $W_U^{\boldsymbol{\cal G}}$ for $n$-qubit unitary $U$ and its negativity also exponentially scales with $n$ as there are ${\cal O}(2^{2n})$ possible phase space points for a $n$-qubit system.

As we can observe from the scaling, the limiting factor of gate merging is the spatial parameter $n$, which stems from the exponential growth of the dimension of Hilbert space by increasing the number of qubits. We find numerically that a practical choice for the spatial parameter $n$ is $n \leq 5$. As this is a fundamental property of a quantum system, a similar issue arises in the robustness measure as evaluating the robustness of ${\cal R}(\ket{T}^{\otimes n})$ and finding its optimal decomposition among ${\cal O}(2^{n^2})$ stabilizer states is in general a challenging task for a large $n$ \cite{howard_2017}.

Due to the computational need to truncate the spatial parameter $n < N$, a question arises of whether there exist new methods of manipulating the circuit frames and further reducing the total negativity, after gate merging is completed. 
We provide a positive answer to this question in the following section, where we describe our second sub-routine, frame optimisation.

\section{Frame optimisation}
\label{sec:opt}

Frame optimisation aims to reduce the total circuit negativity by optimally choosing frames for different circuit components. As we discussed in Section~\ref{subsec:frame_parametrisation}, we can introduce specific frame parametrisations, such as parametrised Wigner frames or rotated Pauli frames, and iteratively choose the frames throughout the circuit.

\begin{subroutine}[H]
\caption{Frame optimisation}
\label{alg:frame_opt}
\begin{algorithmic}[1]
\Statex \textbf{Input:} Quantum circuit $C$ and temporal 
\Statex \hspace{4pt} parameter $\ell$.
\State Determine the total number of frames, $|\boldsymbol{\mathcal{G}}_{\rm opt}|$, 
\Statex \hspace{4pt} in the circuit $C$.
\State Define the set of reference frames, 
\Statex \hspace{4pt} $\boldsymbol{\cal G}_{\rm opt} \gets \{ {\cal G}_1, \dots, {\cal G}_{|\boldsymbol{\mathcal{G}}_{\rm opt}|} \}$.
\State Fix the number of optimisation cycles $c$.
\For{$i = 1, \dots, c$}
    \State Choose a subset $\boldsymbol{\cal G}_{\rm target}^{(i)} \subset \boldsymbol{\cal G}_{\rm opt}$ with at 
    \Statex \hspace{21pt} most $\ell$ frames.
    \State Find a circuit block $\boldsymbol B$ containing the 
    \Statex \hspace{21pt} frames in $\boldsymbol{\cal G}_{\rm target}^{(i)}$.
    \State Find ${\overline{\boldsymbol{\cal G}}}_{\rm target}^{(i)} = {\rm argmin}_{\boldsymbol{\cal G}_{\rm target}^{(i)}} N_{\boldsymbol{B}} \left(\boldsymbol{\cal G}_{\rm target}^{(i)}\right)$.
    \State Update the corresponding frames in $\boldsymbol{\cal G}_{\rm opt}$ 
    \Statex \hspace{21pt} with ${\overline{\boldsymbol{\cal G}}}_{\rm target}^{(i)}$.
\EndFor
\Statex \textbf{Output:} $\boldsymbol{\cal G}_{\rm opt}$.
\end{algorithmic}
\end{subroutine}

In principle, the best strategy in terms of achieving the highest negativity reduction would be to carry out global optimisation over all circuit frames, requiring that the number of parameters to be optimised should scale with the number of qubits $N$ and circuit length $L$.
In this work, we show that a \emph{local} optimisation, with only a fixed number of parameters, is sufficient to achieve considerable negativity reduction and scales only linearly in $N$ and $L$. This optimisation sub-routine is implemented by dividing the circuit into blocks containing at most $\ell$ frames to be optimised, for a fixed temporal parameter $\ell$.

To perform the frame optimisation on a quantum circuit $C$ consisting of an input state $\rho$, a gate sequence $\{U_1,...,U_L\}$ and a measurement effect $E$, we need to start from an initial frame parametrisation.
We denote this parametrisation as $\boldsymbol{\mathcal{G}}_{\rm opt} = \{ {\cal G}_1, \dots, {\cal G}_{|\boldsymbol{\mathcal{G}}_{\rm opt}|} \}$, where $|\boldsymbol{\mathcal{G}}_{\rm opt}|$ is the number of frames to be optimised.
The procedure is outline in Sub-routine~\ref{alg:frame_opt} and explained here.
We take a subset $\boldsymbol{\mathcal{G}}_{\rm target}^{(1)} \subset \boldsymbol{\mathcal{G}}_{\rm opt}$ with up to $\ell$ frames (either sequentially or randomly) and create the block $\mathbf{B}$ of circuit components which are attached to those $\ell$ frames. 
Keeping all other frames in the block $\mathbf{B}$ fixed with the corresponding frames in $\boldsymbol{\calG}_{\rm opt}$, we want to minimise the total negativity of the block $N_{\mathbf{B}}$ over all possible choices for $\boldsymbol{\calG}_{\rm target}$, so that the minimum
\begin{align}
    \min_{\boldsymbol{\calG}_{\rm target}} N_{\mathbf{B}}(\boldsymbol{\calG}_{\rm target}).
\end{align}
occurs at ${\overline{\boldsymbol{\cal G}}}_{\rm target}^{(1)}$ allowing us to update the corresponding frames in $\boldsymbol{\calG}_{\rm opt}$, which is the end of the first cycle in our frame optimisation.
We repeat this process $c$ times by choosing another set of $\ell$ frames as the new $\boldsymbol{\calG}_{\rm target}^{(i)},\ i=1, \dots, c$.
The number of optimisation cycles $c$ can be chosen arbitrarily, for example it can be chosen as $c \geq |\boldsymbol{\mathcal{G}}_{\rm opt}| / \ell$, with the aim of optimising all frames in the circuit at least once.
The order in which frames are optimised can also be chosen arbitrarily and can potentially result in a different overall negativity reduction.

\begin{figure}[t]
    \includegraphics[width=0.97\linewidth]{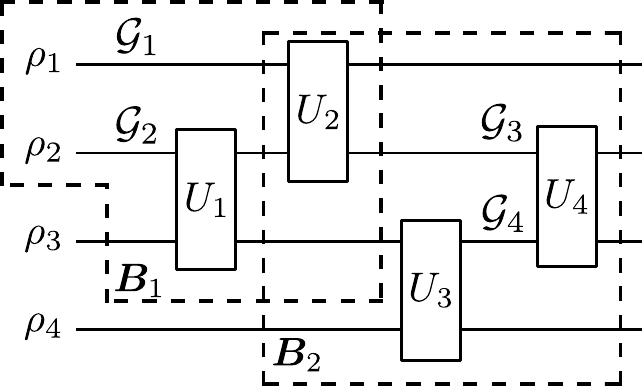}
    \caption{Example of how to form a block when $\boldsymbol{\calG}_{\rm target}$ is given in the case of $n=2$ and $\ell=2$. Only relevant frames are shown. When $\boldsymbol{\calG}_{\rm target}=\{\calG_1,\calG_2\}$, the corresponding block $\mathbf{B}_1$, which contains all circuit elements connected to the frames in $\boldsymbol{\calG}_{\rm target}$, is $\mathbf{B}_1=\{\rho_1,\rho_2,U_1,U_2\}$. When $\boldsymbol{\calG}_{\rm target} = \{ {\calG}_3, {\calG}_4\}$, then the corresponding block is $\mathbf{B}_2=\{U_2,U_3,U_4\}$.
    }
    \label{fig:sketch_block_opt}
\end{figure}

We demonstrate the local frame optimisation method with an example. Let us consider the initial part of a simple general circuit depicted in Fig.~\ref{fig:sketch_block_opt} for the case of $n=2$ and $\ell=2$. To perform the $(i)$-th optimisation cycle we consider $\boldsymbol{\calG}_{\rm target}^{(i)} = \{\calG_1,\calG_2\}$, we consider the corresponding block $\mathbf{B}_1=\{\rho_1,\rho_2,U_1,U_2\}$, which is a set of all circuit components connected to the frames in $\boldsymbol{\calG}_{\rm target}^{(i)}$. 
Then, the explicit optimisation we perform is
\begin{align}
    &\min_{\boldsymbol{\calG}_{\rm target}^{(i)}} N_{\mathbf{B}_1} \left( \boldsymbol{\calG}_{\rm target}^{(i)} \right) = \nonumber\\
    &\min_{\{\calG_1, \calG_2\}} N_{\rho_1} (\calG_1)N_{\rho_2}(\calG_2)N_{U_1}(\calG_2)N_{U_2}(\calG_1),
\end{align}
where $N_X(\calG_X)$ is the negativity of component $X$ as a function of $\calG_X$ with all other frames fixed to the corresponding ones in $\boldsymbol{\calG}_{\rm opt}$. 
As an additional example, we could have considered the set $\boldsymbol{\calG}_{\rm target}^{(i)} = \{{\calG}_3, {\calG}_4\}$ corresponding to the block $\mathbf{B}_2=\{U_2,U_3,U_4\}$ in Fig.~\ref{fig:sketch_block_opt}. 
Then, the block negativity we optimise is
$N_{\mathbf{B}_2}(\boldsymbol{\calG}_{\rm target}^{(i)}) = N_{U_2}({\calG}_3)N_{U_3}({\calG}_4)N_{U_4}({\calG}_3, {\calG}_4)$.

Note that at each optimisation step, previously optimised frames in $\boldsymbol{\calG}_{\rm opt}$ are used in the next optimisation cycle. This ensures that the negativity never increases compared to the initial frame choice $\{ {\cal G}_1, \dots, {\cal G}_{|\boldsymbol{\mathcal{G}}_{\rm opt}|} \}$ between optimisation cycles. 

The presented local optimisation method is efficient in the number of circuit components. Consider an $N$-qubit circuit of length $L$ where each of $L$ gates acts on at most $n$ qubits. Then, there are at most $N+nL$ different frames to be optimised. Since $\ell$ is fixed, each optimisation cycle takes a constant amount of time $O(1)$. Therefore, the frame optimisation of the entire circuit scales as $(N+nL)\times O(1)=O(N,L)$. Note that the exact value depends on truncation parameters $n$ and $\ell$ as well as the specifics of the circuit and its frame parametrisation.

\begin{figure}[!t]
    \includegraphics[width=1
    \linewidth]{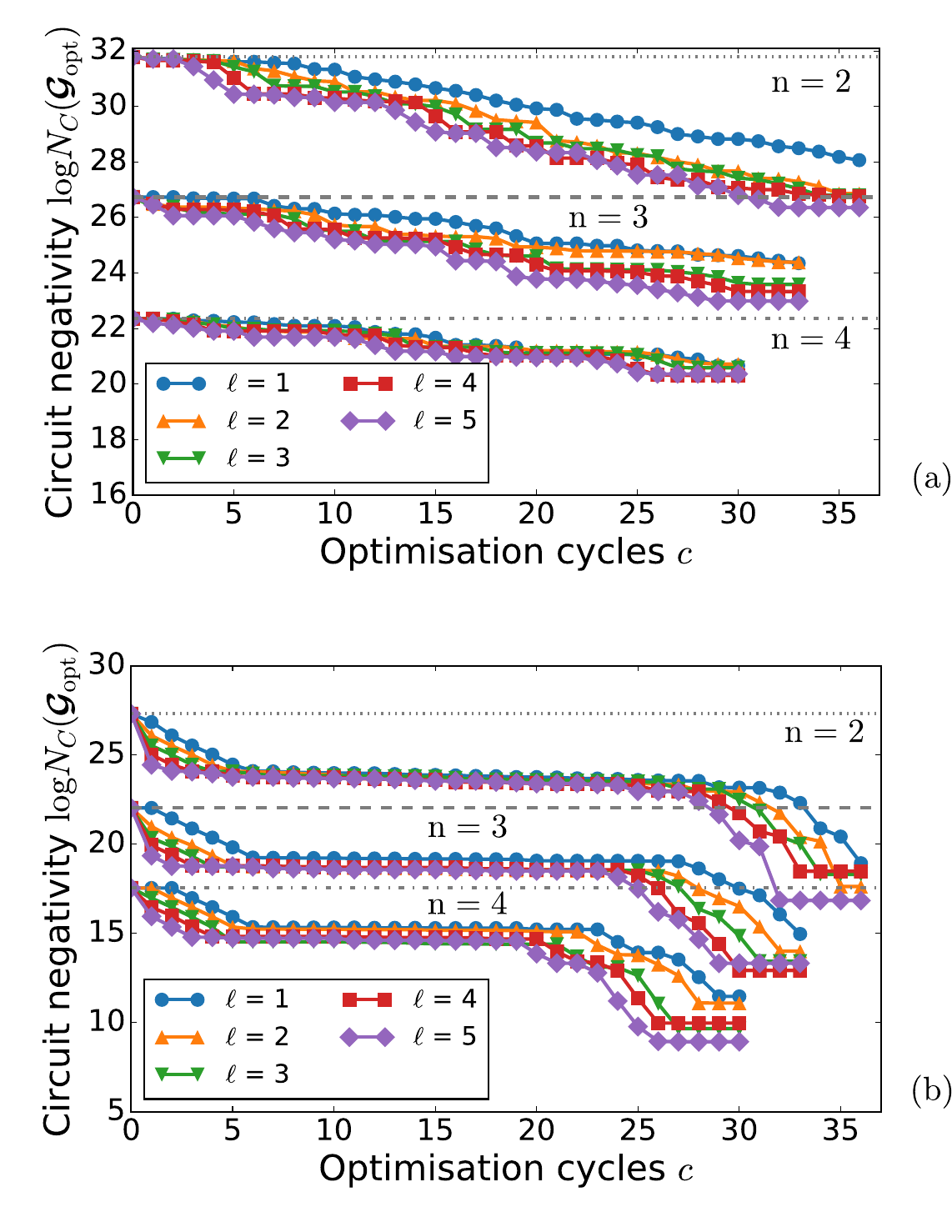}
    \caption{Plots showing negativity reduction of a circuit consisting of 2-qubit Haar-random gates with $N=6$ and $L=15$ after each frame optimisation cycle with different spatial and temporal parameters, $n$ and $\ell$. The optimisation is carried out sequentially from the first frame to the last frame.
    Optimisation is performed via the basin-hopping algorithm as introduced in~\cite{Wales1997}.
    (a) Results after frame optimisation with rotated Pauli frames. The reference frame is the standard Pauli operators. The total negativity continuously decreases as we optimise more frames. (b) Results after frame optimisation with parametrised Wigner frames. The reference frame is the conventional phase-space operators for the Wigner function. The most of negativity reduction occurs near the initial states and the measurements.
    }
    \label{fig:frame_opt_simulation}
\end{figure}

Fig.~\ref{fig:frame_opt_simulation} shows the performance of the frame optimisation for a circuit with $N=6$ and $L=15$ consisting of 2-qubit Haar-random unitaries, which are in general difficult to be simulated with stabiliser-based simulators because they do not admit efficient decompositions. In Fig.~\ref{fig:frame_opt_simulation}(a), we use rotated Pauli frames as our frame parametrisation and initialised each frame in the circuit to the set of standard qubit Pauli operators. In Fig.~\ref{fig:frame_opt_simulation}(b), we choose parametrised Wigner frames as our frame parametrisation and initialise each frame in the circuit to the set of conventional phase-space operators corresponding to $g(\blambda)=1$ (see Sec.~\ref{subsec:examples_frame_para}). We can observe that the largest negativity reduction comes from gate merging with higher $n$, but the frame optimisation also achieves a significant negativity reduction. In general, larger $\ell$ results in lower negativity after optimisation of all frames with fixed $n$. In the case of parametrised Wigner frames, together with gate merging, we could considerably decrease the initial log-negativity from $\sim$27.3 to $\sim$8.9 with truncation parameters $n=4$ and $\ell=5$, which means that we need $\sim 2^{2\times18}$ times less samples to reach a given accuracy for probability estimation.

\begin{figure}[t]
    \includegraphics[width=1
    \linewidth]{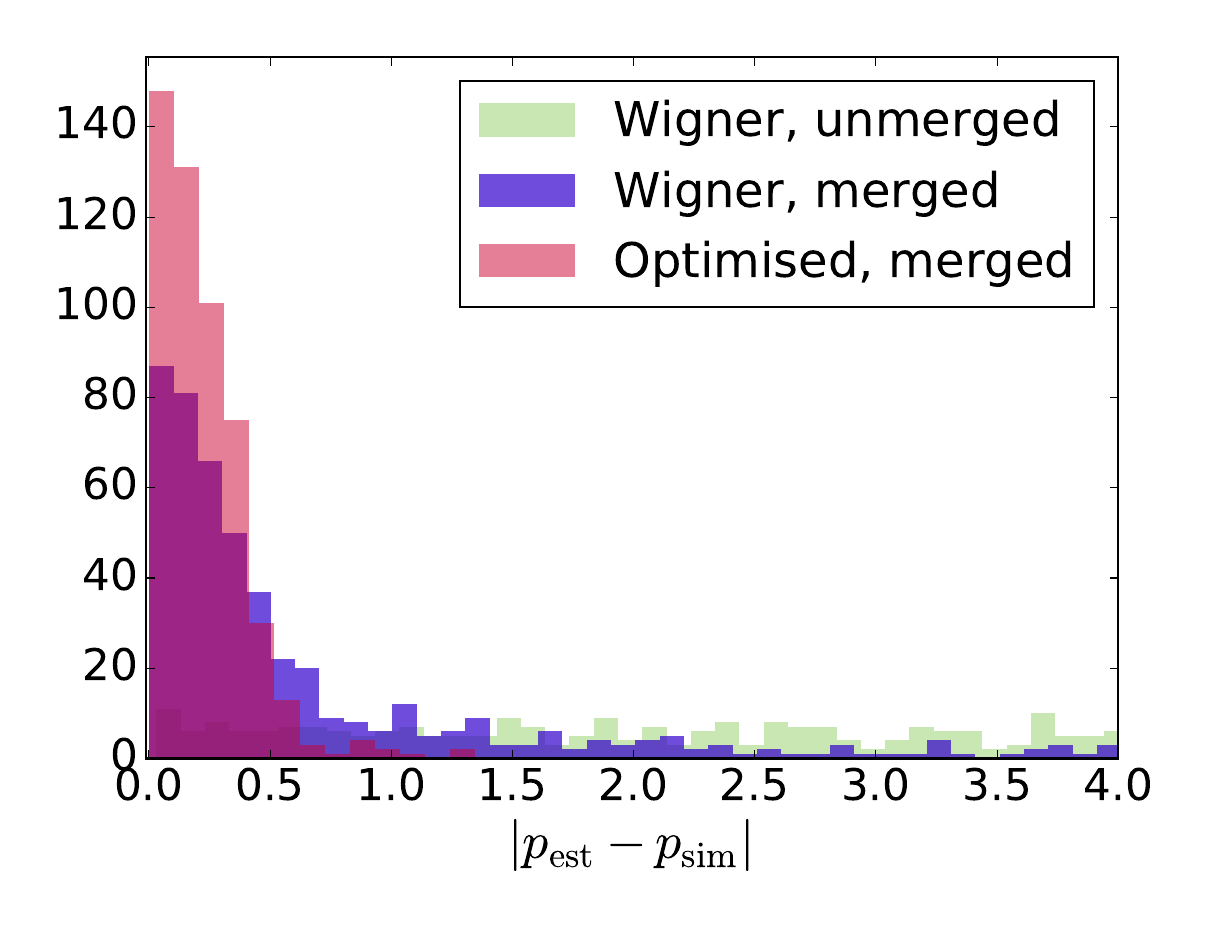}
    \caption{Histograms of the deviation of estimated probability $p_{\rm est}$ from actual outcome probability $p_{\rm sim}$ (as calculated by Qiskit~\cite{Qiskit}) for 500 circuits consisting of 2-qubit Haar-random gates with $N=3$, $L=8$ and $\ell = 1$.
    The number of samples taken for each circuit is $10^6$, which took around 10 seconds on a standard computer.
    The plot shown is truncated at $|p_{\rm est} - p_{\rm sim}| = 4.0$ to demonstrate the advantage of our routines clearly.
    The advantage is amplified as $N$ and $L$ increase.
    }
    \label{fig:sampling_hist}
\end{figure}

We demonstrate the practical significance of our routine, by sampling 500 circuits consisting of Haar-random gates, with the results presented in Fig.~\ref{fig:sampling_hist}.
Unmerged circuits represented entirely by Wigner frames do not show any signs of convergence to the actual probability distribution.
Merged circuits clearly converge a lot better, especially when their frame representation is optimised.

\section{Conclusion}
\label{sec:summary}

We introduce two classical sub-routines, gate merging and frame parametrisation, which reduce the total negativity in the quasi-probability representation of a quantum circuit, hence leading to sampling overhead reduction.
We emphasise that our methods are very general; they are not restricted to specific choices of frames or frame parametrisations, and can be applicable to any circuit independently of generating gate sets or the purity and dimension of its input qudits.
Both sub-routines are efficient in the sense that the runtime scales polynomially in the circuit size $N$ and number of gates $L$.

We numerically demonstrate that both methods improve the exponential scaling of the circuit negativity by testing them on Clifford+$T$ circuits and circuits with Haar-random gates.
Specifically, gate merging is shown to compete on average with the quasi-probability simulators based on dyadic frames and the robustness of magic.
Frame optimisation can further compliment gate merging in reducing negativity, when merging gates in the circuits is no longer practical due to the growing size of the gates.

A clear direction for our work is to improve the classical optimisation performed for the frame representation.
Our parametrisation resembles variational techniques used in near-term quantum algorithms~\cite{cerezo_variational_2021}, although our cost function, circuit negativity, is calculated classically.
Our optimisation could therefore potentially benefit by research on variational techniques, such as identifying ``good'' circuit-inspired ansatze for initialising frames or investigating barren plateaus in order to improve optimisation convergence.
Such methods could shed light on what families of circuits are hardest to sample from using quasi-probability techniques.

One can also investigate the possibility of performing frame optimisation analytically, at least for particular classes of quantum circuits. 
Additional assumptions will likely be required for the circuit structure, but finding optimal frames analytically would eliminate the hidden constant runtime costs of ``black-box'' classical algorithms currently employed for the optimisation.
For example, it would be particularly useful to investigate the existence of a finite set of frames as a function of circuit components resulting in minimum negativity for Clifford+$T$ circuits.

\section*{Acknowledgements}

This work is supported by the KIST Open Research Program and the UK Hub in Quantum Computing and Simulation, part of the UK National Quantum Technologies Programme with funding from UKRI EPSRC grant EP/T001062/1.
NK is supported by the EPSRC Centre for Doctoral Training on Controlled Quantum Dynamics.
HK is supported by the KIAS Individual Grant No. CG085301 at Korea Institute for Advanced Study.
HHJ is supported by the Centre for Doctoral Training on Controlled Quantum Dynamics funded by the EPSRC (EP/L016524/1).
DJ is supported by the Royal Society and a University
Academic Fellowship.
MSK acknowledges the Samsung GRC grant.

\section*{Code availability}

The code that implements the sub-routines and supports the presented numerical findings can be accessed at NK's GitHub repository: \\ \url{https://github.com/nkoukou/parameterised_negativity}.

\bibliographystyle{apsrev4-1}
\bibliography{algo}

\end{document}

%% file: algo_preamble.tex
\usepackage{amsmath, amssymb, amsfonts, amsthm}
\usepackage{bm, bbm, bbold, physics, mathtools}
\usepackage{graphicx, subfigure, epstopdf}
\usepackage{dcolumn, array, enumerate}
\usepackage{color, xcolor, lipsum, xifthen}
\usepackage{verbatim, hyperref}
\usepackage{algorithm}
\usepackage[noend]{algpseudocode}
\usepackage[normalem]{ulem}

\hypersetup{
	colorlinks=true,  
	linkcolor=blue,   
	citecolor=blue,   
	urlcolor=blue     
}



\newtheorem{theorem}{Theorem}
\newtheorem{observation}{Observation}

\newfloat{subroutine}{htbp}{loa}
\floatname{subroutine}{Sub-routine}

\renewcommand{\tr}{{\rm{tr}}}

\def\bmth{\bm{\theta}}

\def\bra#1{\langle #1|}
\def\ket#1{\left|#1 \right>}

\def\tr{\mbox{Tr}}

\newcommand{\calG}{\mathcal{G}}

\def\bomega{{\boldsymbol{\omega}}}
\def\blambda{{\boldsymbol{\lambda}}}

\def\sign{{\rm Sign}}

\renewcommand{\cal}[1]{\mathcal{#1}}

\renewcommand{\tr}{{\rm{tr}}}

\def\bmth{\bm{\theta}}

\def\bra#1{\langle #1|}
\def\ket#1{\left|#1 \right>}

\def\tr{\mbox{Tr}}

\def\bomega{{\boldsymbol{\omega}}}
\def\blambda{{\boldsymbol{\lambda}}}

\renewcommand{\cal}[1]{\mathcal{#1}}

\newcommand{\id}{\mathbbm{1}}
\renewcommand{\tr}{{\rm{tr}}}

\def\bmth{\bm{\theta}}

\def\bra#1{\langle #1|}
\def\ket#1{\left|#1 \right>}

\def\tr{\mbox{Tr}}

\def\bomega{{\boldsymbol{\omega}}}
\def\blambda{{\boldsymbol{\lambda}}}

\def\sign{{\rm Sign}}

\renewcommand{\cal}[1]{\mathcal{#1}}

\renewcommand{\tr}{{\rm{tr}}}

\def\bmth{\bm{\theta}}

\def\bra#1{\langle #1|}
\def\ket#1{\left|#1 \right>}

\def\tr{\mbox{Tr}}

\def\bomega{{\boldsymbol{\omega}}}
\def\blambda{{\boldsymbol{\lambda}}}